\documentclass[preprint,12pt]{amsart}
\usepackage[pagewise]{lineno}

\usepackage{graphicx}
\usepackage{amsmath} 
\usepackage{amssymb}
\usepackage{amsthm}
\usepackage{booktabs}
\usepackage[update,prepend]{epstopdf}
\usepackage[slantedGreek]{mathptmx}
\usepackage{url}
\usepackage{bbm}
\usepackage{accents}
\usepackage{subfigure}

\newlength{\dhatheight}

\newcommand{\supp}{\mbox{\rm supp}\; }
\newcommand{\grad}{\nabla{}}
\renewcommand{\div}{\mbox{\rm div}}

\newcommand{\xb}{\bar{x}}

\newcommand{\br}{{\rm BR}}
\newcommand{\xh}{\hat{x}}
\newcommand{\ab}{\bar{a}}
\newcommand{\bb}{\bar{b}}
\newcommand{\rb}{\bar{r}}
\newcommand{\esssup}{\mbox{\rm ess\,sup}}

\newtheorem{theorem}{Theorem}
\newtheorem{lemma}{Lemma}

\newtheorem*{theorem*}{Theorem}

\begin{document}


\title{A model of cultural evolution in the context of strategic conflict}
\author{Misha Perepelitsa }

\date{\today}
\address{
mperepel@central.uh.edu\\
Department of Mathematics\\
University of Houston\\
4800 Calhoun Rd. \\
Houston, TX.}

\maketitle

\begin{abstract}
We consider a model of cultural evolution 
for a strategy selection in a population of individuals who interact in a game theoretic framework. The evolution combines individual learning of the environment (population strategy profile), reproduction, proportional to the success of the acquired knowledge, and social transmission of the knowledge to the next generation. A mean-field type equation is derived that describes the dynamics of the distribution of cultural traits, in terms of the rate of learning, the reproduction rate and population size. We establish global well-posedness of the initial-boundary value problem for this equation
and give several examples that illustrate the process of the cultural evolution for some classical games.
\end{abstract}






\begin{section}{Introduction}
Evolutionary game theory, pioneered by Maynard Smith and Price \cite{Maynard-Price} is a powerful tool  that explains dominance of some behavioral traits as  being uninvadable  by other traits in the competition for Darwinian fitness points, when  fitness is frequency dependent. A deterministic dynamic process that selects a stable behavioral traits can be described by the replicator equation, see Taylor and Jonker \cite{Taylor-Jonker}, Hofbauer et al. \cite{Hofbauer79}, Zeeman \cite{Zeem}.

The replicator equation also governs the dynamics of reinforcement learning in repeated play of a game, see Borgers and Sarin \cite{Borgers-Sarin},  Fudenberg and Levine \cite{FudenbergLevine}, Krishnedu et al. \cite{NowakKrishnedu}, Perepelitsa \cite{Perepem1}.

Learning in games is an integral part of game theory that goes back to works of Robinson \cite{Robinson} and Shapley \cite{Shapley}. One of its mainstays is fictitious play or statistical learning. The learning by fictitious play in large populations can be described by an ODE, called best-response equation, see Gilboa and Matsui \cite{GM}, Gaunersdorfer and Hofbauer \cite{GH}, Hofbauer \cite{Hofbauer} and Hofbauer and Sigmund\cite{Hofbauer-Sigmund}. The best-response equation describes  changes in the mean statistical prior about the opponent actions, and its stationary points are Nash equilibria.

In this paper we consider an evolutionary process that combines the concept of ``the survival of fittest'' from biological evolution with individual learning through fictitious play, when a state of learning  is socially transmitted to the next generation of players.
The examples of this types of processes are furnished by the cultural evolution theory.
 
Consider a cartoonish scenario of cultural evolution. Lets say there is an island populated by pedantical statisticians capable of asexual reproduction. Statisticians wonder aimlessly around an island, meeting each other occasionally for a round of a game (symmetric game with a finite set of strategies). Each carries a ledger where he/she carefully marks how many times  the opponent played a particular strategy (opponents as indistinguishable). To select a strategy, each of them uses the ``sacred book of rules'' (best response in a game) that prescribes what to do given the current count from his/her ledger. The book has a dual purpose of  settling the outcome of each play of the game (payoffs), and the players collect certain amount of fitness points from each play.  
From time to time statisticians reproduce at the rate proportional to their accumulated fitness. When that happens, they solemnly pass an exact copy of their ledgers to each of the offspring, who carry on with it in the same manner.

We say that this model is a form of cultural evolution because it is characterized by social transmission of traits (inheritance of knowledge) and individual learning as adaptation mechanisms,  see Hoppitt and  Laland \cite{Hoppitt-Laland}, Richerson and Boyd \cite{Richerson-Boyd}. Note, that in this case, social transmission and learning change population strategy profile (environment), which determines, in its turn, the degree of  success of a cultural trait, rendering  the problem nonlinear.

The main parameters of the problem are the rate of learning, the rate of reproduction
and population size. Additional information might be needed to completely specify the problem. For example, if the island is not big, we may assume that the frequency with which inhabitants meet and play the game increases with the population size. Another scenario is an infinite island which allows inhabitants to spread, no matter how many of them are there, so that the interaction frequency is constant. 

The goal of the paper is to develop a mathematical  model that 
takes as input the initial distribution of cultural traits in a population, the above mentioned rates of learning and reproduction, and outputs the distribution of cultural traits at any moment $t$ in future.  As we will see from the examples of section \ref{examples}, it is essential for an accurate description of the dynamics that the model specifies the whole distribution of traits and not just some statistical averages, such as the mean and the variance.    The model is derived as a mean-field approximation of the distribution density of a Markov process describing the interaction of agents. The equation is of kinetic type with  non-linear kinetic velocities. Due to the discontinuities of the best-response function, solutions of this equation are intrinsically weak. Our main result, stated in section \ref{model}, establishes global well-posedeness of the initial-boundary value problem for this equation.

In section \ref{examples} we discuss two examples that illustrate the dynamics of  this model of cultural evolution.
In the first, we consider Hawk-Dove-Retaliator game with two evolutionary stable strategies: $\frac{1}{2}H +\frac{1}{2}D$ and $R.$ 
These two strategies are the only  asymptotically stable points for the replicator and best-response equations. The phase, however, are different, with the basin of attraction for $\frac{1}{2}H +\frac{1}{2}D,$ for the replicator equation, being strictly included in the basin of attraction  for the best-response dynamics. As a result, there are initial conditions for the distribution of cultural traits which evolve to Retaliator when the rate of learning is low, but  proceed to  $\frac{1}{2}H +\frac{1}{2}D$ when learning rate is increased. It also means that if the biological evolution proceeds to the mix of Hawk and Dove, it can not be averted to anything else by learning.  Another  interesting property of this process is a sharp change in the environment (population strategy profile) when a subpopulation continuously  transitions from one decision polygon to another.

In the second example, of Rock-Paper-Scissors game, we show how exponentially growing heterogeneous population can lock the cultural evolution in a suboptimal pure strategy, in contrast to both, the dynamics of the replicator and best-response equations.
 
In general, determining asymptotic state for this type of evolution for an arbitrary game   is problematic due to complicated dynamics and  absence of entropy functionals. It can be done in some cases, at least partially,  as in the model with zero reproduction rate. For that model, we show that if two statistical averages, the mean prior and mean best-response converge to some values (not necessarily the same), then the prior of every agent in the population converge to  a Nash equilibrium of the game.

\end{section}

\begin{section}{Model}
\label{model}
We consider a series of plays of a symmetric 2-player game between randomly selected agents in a large population. There are $d$ strategies available to agents and the payoffs are given by matrix $A=\{a_{ij}\}_{i,j=1}^d,$  which we assume to have non-negative entries. The game defines multi-valued best response function $\br(p)\,:\,S_{d-1}\to \mathcal{P}(S_{d-1}),$ where $S_{d-1}$ is $d-1$ dimensional simplex. We will also use its single valued representative $b(p) \in \br(p/\sum_i p_i).$ We refer the reader to Appendix section \ref{best:response} for details.
We start with the case when there is no reproduction  
and the population stays at the same level $N$.

We will record the change of the state of agents that occur at discrete epochs, labeled by $t.$
The state of agent $i$ at epoch $t$ is a $d+1$ dimensional vector
\[
X^t_i{}={}\left(P^t_i,S^t_i\right),
\]
where
\[
P^t_i {}={}(p^t_{i,1},...,p^t_{i,d}),
\]
is a vector of learning priors (unscaled) and $S^t_i$ is an averaged, accumulated fitness.

An interaction is a round of the game between to random agents, say $i$ and $j,$  who play according to their priors $P^t_i$ and $P^t_j,$ that is, using their best response strategies. 
Based on that, they earn fitness points and update the learning priors. To describe the  update rule we will use the following  parameters: $h$ -- the characteristic learning increment, $\mu h$ -- characteristic fitness increment, and $\delta$ -- time increment.
Thus we assuming the learning and fitness increments are of the same order, but not necessarily equal. The rule takes the form
\begin{eqnarray*}
P^{t+\delta}_i &=& P^t_i + hb(P^t_j)\\
 P^{t+\delta}_j &=& P^t_j + hb(P^t_i)\\
 S^{t+\delta}_i &=& (1-\mu h)S^t_i {}+{}\mu h a(b(P^t_i),b(P^t_j))\\
 S^{t+\delta}_j &=& (1-\mu h)S^t_j {}+{}\mu h a(b(P^t_j),b(P^t_i))
\end{eqnarray*}
where $a(b(P^t_i),b(P^t_j)){}={}\sum_{k,l}a_{ij}b(P^t_{i,k})b(P^t_{j,l})$ is 
the fitness earned by agent $i.$
In this formulas the fitness is averaged over the history of payoffs, so that it can not grow without a bound. One can think of $\mu$ as being a recency parameter. Large values of $\mu$ put more weight on more recent payoffs.

Our goal here is to derive an approximate equation for $f(p,s,t)$ -- the density of the distribution of agents over the space of learning priors and fitness $(p,s)\in \mathbb{R}^{d}_+\times\mathbb{R}_+.$ In the following derivation we use the convention that
\[
x_i = (p_i,s_i),\quad \xb {}={}(x_1,...,x_N)\in (\mathbb{R}^d_+\times\mathbb{R}_+)^N,
\]
where $\xb$ parametrized the state of the whole population.

Let $w(\xb,t)$ be the density of the distribution of priors and fitness for the whole population. This function implicitly  depends on parameters such as $h_1,h_2,$ and $\mu,$ but we suppress them from notation for convenience of presentation.   The update rule can be expressed as a moment relation with a test function $\phi,$
\[
\int \phi(\xb)w(\xb,t+\delta)\,d\xb {}={}\sum_{i\not=j}(N(N-1))^{-1}\int \phi(\xb)\Big|_{x_i=\xh_i \atop  x_j=\xh_j}w(\xb,t)\,d\xb,
\]
where $\xh_i{}={}(p_i + hb(p_j),\,  (1-\mu h)s_i {}+{}\mu h a(b(p_i),b(p_j))),$
ans symmetrically for $\xh_j.$ The last equation can be written as
\begin{equation}
\label{eq:int1}
\int \phi(\xb)[w(\xb,t+\delta)-w(\xb,t)]\,d\xb {}={}\sum_{i\not=j}(N(N-1))^{-1}\int [\phi(\xb)\Big|_{x_i=\xh_i \atop  x_j=\xh_j}- \phi(\xb)]w(\xb,t)\,d\xb.
\end{equation}

Function $f(x,t),$ where $x=(p,s)$ is related to the multi-agent distribution $w(\xb,t)$ through the rule:
\[
f(x,t){}={}\sum_k N^{-1}\int w(\xb)\big|_{x_k=x}\,d\xb_k,\quad x\in\mathbb{R}^d_+\times\mathbb{R}_+,
\]
where $\xb_k$ is a $(d+1)(N-1)$ dimensional vector of all coordinates, $x_1,...,x_N,$  excluding $x_k.$ This is one-particle distribution function. In the formulas to follow we need to use two-particle distribution function 
\[
g(x,y,t){}={}\sum_{i\not=j} (N(N-1))^{-1}\int w(\xb)\big|_{x_i=x,\, \\x_j=y}\,d\xb_{ij},
\]
where $\xb_{ij}$ is the $(d+1)(N-2)$ dimensional vector of all coordinated excluding $x_i$ and $x_j.$

Function $g$ is symmetric in $(x,y)$ and is related to $f$ by the formulas
\[
f(x,t){}={}\int g(x,y,t)\,dx{}={}\int g(x,y,t)\,dy.
\]
The moments of function $f$ and $g$ are computed from the  moments of $w:$
\[
\int \psi(x)f(x,t)\,dx{}={}\sum_k N^{-1}\int \psi(x_k)w(\xb)\,d\xb,
\]
and 
\[
\int \omega(x,y)g(x,y,t)\,dxdy{}={}\sum_{i\not=j} (N(N-1))^{-1}\int \omega(x_i,x_j)w(\xb)\,d\xb.
\]
Now, we use \eqref{eq:int1} to obtain an integral equation of the change of function $f.$ For that select $\phi(\xb){}={}\psi(x_k),$ sum over $k$ and take average. We get
\begin{multline}
\int\psi(x)[f(x,t+\delta)-f(x,t)]\,dx\\
{}={}N^{-1}\sum_{k}\sum_{i\not=j} (N(N-1))^{-1}
\int[\psi(x_k)\Big|_{x_i=\xh_i \atop  x_j=\xh_j}-\psi(x_k)]w(\xb,t)\,d\xb \\
{}={} N^{-1}\sum_{i\not=j} (N(N-1))^{-1}\left(\int[\psi(\xh_i)-\psi(x_i)]w(\xb,t)\,d\xb\right.\\
\left. 
+\int[\psi(\xh_j)-\psi(x_j)]w(\xb,t)\,d\xb\right)\\
{}={}\frac{2}{N}\iint[\psi(\xh) - \psi(x)]g(x,y,t)\,dxdy,
\end{multline}
where $x=(p,s),\,y=(p',s')$ and
\[
\xh{}={}(p + hb(p'),\,  (1-\mu h)s {}+{}\mu h a(b(p),b(p'))).
\]

To proceed to, we make an assumption of statistical independence of the states 
of two randomly selected agents:
\[
g(x,y,t){}={}f(x,t)f(y,t).
\]
The plausibility of this condition is partially justified if the population is large, so that same agents are rarely matched together,  and the information about the interaction is not shared between other agents.

Then, expanding $\psi(\xh)$ in Taylor series and integrating by parts, we obtain
\begin{multline}
\int\psi(x)[f(x,t+\delta)-f(x,t)]\,dx\\{}={}\frac{2}{N}\iint \psi(p,s)\div_{p,s}\left(
(hb(p'),\mu hs-\mu ha(b(p),b(p'))) f(x,t)\right)f(y,t)\,dydx \\
{}+{}O(h^2)\\
{}={}
\frac{2}{N}\int \psi(p,s)\div_{p,s}\left(
(h\bb(t),\mu hs-\mu h\ab(b(p),t)) f(x,t)\right)\,dx \\
{}+{}O(h^2),
\end{multline}
with the mean best response
\begin{equation}
\label{bb}
\bb(t) = \iint b(p) f(p,s,t)\,dpds,
\end{equation}
and the mean fitness for using strategy $b(p):$
\begin{equation}
\label{ab}
\ab(b(p),t) = \iint a(b(p),b(p'))f(p',s',t)\,dp'ds'{}={}\sum_{i,j} a_{ij}b_i(p)\bb_j(t).
\end{equation}
Dividing equation by $\delta$ and ignoring higher order terms we arrive at Fokker-Planck equation for density $f(p,s,t):$
\begin{equation}
\label{FK:1}
\partial_t f{}+{}\frac{2h}{N\delta}\div_p\left(\bb(t)f\right) {}+{}
\frac{2h\mu}{N\delta}\partial_s\left((\ab(b(p),t)-s)f\right){}={}0.
\end{equation}

In passing from a discrete to continuous time model we are assuming $\delta,h$ are small, $N$ is large, so that ratios 
\begin{equation}
\label{alphas}
\alpha_p{}={}\frac{2h}{N\delta},\quad \alpha_s{}={}\frac{2\mu h}{N\delta}
\end{equation}
are of finite order. Note that $(N\delta)^{-1}$ can be interpreted as a number of interactions per agent, per unit of time. We're assuming that this number is large and inversely proportional to the characteristic learning and fitness increment $h.$

Now we extend the model to variable populations, by allowing agents to reproduce at the rate proportional their level of fitness. At this point we proceeding heuristically, leaving out the details of the derivation. 

With reproduction, the Fokker-Planck equation must be  appended by a source term
proportional to $(s-\rb(t))f(p,s,t)$ on the right-hand side of \eqref{FK:1}, where $\rb(t)$ is mean population fitness
\begin{equation}
\label{rb}
\rb(t){}={}\sum_{i,j}a_{ij}\bb_i(t)\bb_j(t).
\end{equation}
Mean population size  $N=N(t),$ which is determined from the equation
\begin{equation}
\label{Pop}
\frac{1}{N}\frac{d N}{dt}{}={}\alpha \iint sf(p,s,t)\,dpds,
\end{equation}
where $\alpha$ is reproduction rate. Moreover rates $\alpha_p,$ $\alpha_s$ are variable and depend on $N=N(t).$
The final  model reads:
\begin{equation}
\label{FK:2}
\partial_t f{}+{}\alpha_p\div_p\left(\bb(t)f\right) {}+{}
\alpha_s\partial_s\left((\ab(b(p),t)-s)f\right){}={}\alpha (s-\rb(t))f,
\end{equation}
with $\alpha_p,\alpha_s,\bb(t),\ab(b(p),t)$ and $\rb(t)$ given by \eqref{alphas}, \eqref{bb}, \eqref{ab}, and \eqref{rb}, respectively. Note that equations \eqref{Pop} and \eqref{FK:2} are coupled through formulas \eqref{alphas}.

\subsection{Singular limit of recency parameter $\mu.$} In the reproduction scenario described by \eqref{FK:2}, children acquire not only knowledge $p$ of parents but also their averaged, accumulated fitness $s.$  Hypothetically, this might be a valid assumption in some situations, however, it seems more relevant to consider the case that it is only knowledge $p$ that eventually determines the fitness of offspring. This can easily be achieved in the framework of models \eqref{alphas}-\eqref{FK:2} be taking the limit of $\mu\to\infty$ ($\alpha_s\to \infty$), which overweights the stimulus obtained from  recent encounters. 
For the derivation of the new model we proceed informally. Dividing equation \eqref{FK:2} by $\alpha_s$ and passing to the limit, we get 
\[
\partial_s ( (\ab(b(p),t)-s)f){}={}0.
\]
Since $f$ is non-negative, this equation can be true only if for all $p\in\mathbb{R}^d_+,$ and $t>0,$ $f$ is a delta-function concentrated on value $\ab(b(p),t):$
\[
f(p,s,t) = \delta(s-\ab(b(p),t)).
\]
That is,  fitness equals to the expected payoff for an agent using strategy $b(p)$ against the population strategy profile $\bb(t)$:
\[
s = \ab(b(p),t){}={}\sum_{ij}a_{ij}b_i(p)\bb_j(t).
\]
Now, the dimension of the problem can be reduced, as we can integrate \eqref{FK:2} in $s,$ and find an equation for moment $\int_{-\infty}^\infty f(p,s,t)\,ds,$ which, with slight abuse of notation, we still call $f(p,t).$  The equation reads:
\begin{eqnarray}
\label{FK:3}
\partial_t f{}+{}\alpha_p\div_p \bb(t)f {}&=&{}\alpha (\ab(b(p),t) - \rb(t))f \notag\\
 &=& \alpha \left(\sum_{ij}a_{ij}b_i(p)\bb_j(t) {}-{}\sum_{ij}a_{ij}\bb_i(t)\bb_j(t)\right)f.
\end{eqnarray}
This is the equation of our main interest, for which we will establish global well-posedness. Before we switch to the mathematical analysis, we mention a special case with zero reproduction $\alpha=0.$

To complete the mathematical setup for equations \eqref{FK:3} and \eqref{SL:1} it remains to add the initial conditions for the population size $N(0)=N_0,$ for the density
\begin{equation}
\label{IC}
f(p,0){}={}f_0(p),\quad p\in\mathbb{R}^d_+.
\end{equation}
and boundary conditions (zero influx of probability):
\begin{equation}
\label{BC}
f(p,t){}={}0,\quad p\in\partial\mathbb{R}^d_+,\,t\geq0.
\end{equation}
Note that  velocity vector $\bb(t)$ is always directed into $\mathbb{R}^d_+,$ and the problem is not over-determined.

\subsection{Fictitious play in large populations} The model becomes particularly simple:
\begin{equation}
\label{SL:1}
\partial_t {f} {}+{}\alpha_p\div(\bb(t){f}){}={}0,
\end{equation}
with the mean best response 
\begin{equation}
\label{bb}
\bb(t){}={}\int b(p){f}(p,t)\,dp.
\end{equation}
Using equation \eqref{SL:1} we can compute the equation for the mean empirical frequencies vector $P(t):$
\begin{equation}
\label{P}
\frac{dP_i}{dt}{}={}\int \frac{1}{\sum_j p_j}\left(\bb_i(t) - \frac{p_i}{\sum_j p_j}\right)\tilde{f}(p,t)\,dp,\quad i=1..d,
\end{equation}
since $\sum_j \bb_j(t){}={}1.$  If one postulates that all agents have the same, or approximately the same, priors 
\begin{equation}
\label{H:delta}
p(t)=(P_1(t),..,P_d(t)),
\end{equation}
then the above equation reduces  to a variant of  the best response  dynamics equation:
\begin{equation}
\label{BRD}
\frac{dP_i}{dt}{}={}\frac{1}{\sum_j P_j(t)}\left( \bb_i(P) - P_i  \right),\quad i=1..d.
\end{equation}
Notice, also, the  positive factor on the right-hand side of the equation. For a learning processes in which priors become large,  the learning rate slows down. 

\subsection{Relation to the replicator equation}
With zero learning rate $\alpha_p=0$ model \eqref{FK:3} is simply the replicator equation written in terms of the distribution function $f(p,t).$ Indeed, in this case each agent uses a fixed strategy $b(p),$ so that the population is split into at most $d$ groups, each using a particular strategy, and each reproducing at the rate proportional to the averaged fitness 
obtained from interacting with whole population. Formally, one obtains the system of replicator equations by integrating  \eqref{FK:3} over sets $\{p\,:\, b(p) = e_k\},$ $k=1..d.$

\subsection{Existence of weak solutions}
In this section we establish our main result, theorem \ref{th:1}.
Let $\Omega=\mathbb{R}^d_+,$ and $C^1_0(\Omega)$ be a space of continuously differentiable functions with compact support in $\Omega.$ We adopt standard notation for $L^p(\Omega)$ spaces and the space of functions of locally bounded variation   $BV_{loc}(\Omega).$ The latter consists of all measurable and locally integrable functions $f$ such that for any ball $B_r,$
\[
\|f\|_{TV(B_r\cap\Omega)}{}={}\sup \left\{
\int_{B_r\cap\Omega}f\div \psi\,dp\,:\, \psi \in C^1_0(B_r\cap \Omega),\, \sup_p|\psi|\leq1
\right\}<+\infty.
\]
For such functions, the distributional derivative $\partial_{p_i} f,$ $i=1..d,$ is a signed Radon measure. One can find the information on these spaces and the results from functional analysis that we use below, for example, in a  book by Brezis \cite{Brezis}.

\begin{theorem}
\label{th:1}
Let $f_0\in C^1_0(\Omega)$ be a non-negative function with unit mass. There is a unique weak solution $f$ of \eqref{FK:3}, \eqref{IC}, \eqref{BC} such that 
\[
f\in {}C([0,T];L^1(\Omega))\cap L^\infty([0,T];BV(B_r\cap\Omega)),\quad \forall r,T>0.
\]
For any $t>0,$ $f(p,t)\geq 0, $ a.e. in $\Omega$ and $\int f(p,t)\,dp{}={}1.$
\end{theorem}
\begin{proof}
From the definition of function $b(p)$ and properties of $\br(p)$ it follows that for any ball $B_r,$  $b(p)$ has finite total variation on $B_r\cap \Omega,$ and there is $C=C(r),$ but not depending on the center of the ball, such that
\begin{equation}
\label{b:tv}
\|b\|_{TV(B_r\cap\Omega)}{}\leq{} C.
\end{equation}
Equation \eqref{FK:3} can be written in non-conservative form as
\begin{equation}
\partial_t f{}+{}\bb(t)\grad f{}={}\left(\sum  a_{ij}(b_i(p)\bb_j(t)-\bb_i(t)\bb_j(t)) \right)f,
\end{equation}
where for simplicity we set $\alpha=1.$ Given a continuous function $\bb(t)$ we solve this equation by the method of characteristics. For a mapping $X^t\,:\,\mathbb{R}^d\to \mathbb{R}^d,$ defined as
\[
X^t(p) = p + \int_0^t\bb(\tau)\,d\tau,
\]
$f$ is expressed through the formula
\begin{equation*}
f(X^t(p),t){}={}f_0(p)\exp\left\{
\int_0^t\sum a_{ij}[b_i(X^\tau(p))\bb_j(\tau) - \bb_i(\tau)\bb_j(\tau)]\,d\tau
\right\},
\end{equation*}
or as
\begin{multline}
\label{f:formula}
f(p,t){}={}f_0(p-\int_0^t\bb(\tau)\,d\tau)\\
\times\exp\left\{
\int_0^t\sum a_{ij}[b_i(p - \int_\tau^t\bb_i(s)\,ds)\bb_j(\tau) - \bb_i(\tau)\bb_j(\tau)]\,d\tau
\right\}.
\end{multline}

Let $g\in C([0,T];L^1(\Omega))$ be a non-negative function such that $g(p,0)=f_0(p),$
and $\int g(p,t)\,dp=1,$ for all $t\in[0,T].$ We denote this subset of functions as $K.$ It is a closed, convex subset of $C([0,T];L^1(\Omega)).$ Let
\[
\bb_g(t){}={}\int b(p)g(p,t)\,dp,
\]
and define map $f=\mathcal{L}(g)$ by evaluating \eqref{f:formula} with $\bb=\bb_g.$
Notice that due to assumptions on $g,$ $\sup_t|\bb_g(t)|\leq 1.$ It follows that
\[
\sup_{p,t}f(p,t)\leq{}e^{CT}\sup_p f_0(p),
\]
for some $C>0$ independent of $g,$ and $\int f(p,t)\,dp{}={}1.$ Moreover,  the following lemma  holds
\begin{lemma} For any $r>0,$
$f\in L^\infty((0,T);BV(B_r\cap\Omega)),$ and there is $C=C(r,T),$ independent of $g$ such that 
\[
\esssup_t\|f(\cdot,t)\|_{TV(B_R\cap\Omega)}\leq C(r,T).
\]
\end{lemma}
\begin{proof}
Recall that $b(p)$ is a function of finite total variation that verifies estimate \eqref{b:tv}.
Differentiating  
\eqref{f:formula} in $p_k,$ and using the chain rule we  find that for any ball $B_r,$
\begin{multline}
\int_{B_r}|\partial_{p_k}f|\,dp{}\leq{} C(T)\int_\Omega|\partial_{p_k} f_0|\,dp\\
{}+{} C(T)\sup f_0 \int_0^t \int_{B_r-\int_\tau^t \bb_g(s)\,ds} \sum_i |\partial_{p_k}b_i(p)|\,dpd\tau
\\
{}\leq{}C(T)\int_\Omega|\partial_{p_k} f_0|\,dp
{}+{} C(r,T)\sup f_0{}\leq{}C(r,T),
\end{multline}
where $|\partial_{p_k}b_i(p)|$ is a Borel measure.
\end{proof}
Using the argument of the last lemma one easily verifies that $f$ is Lipschitz continuous in time with values in $L^1(\Omega):$
\begin{lemma}
Let $t$ and $t+\delta\in [0,T].$ Then, there is $C=C(T),$  independent of $g,$ such that
\begin{equation}
\int |f(p,t+\delta)-f(p,t)|\,dp {}\leq{} \|f_0\|_{C^1(\Omega)}C\delta. 
\end{equation}
\end{lemma}

From the properties of $f=\mathcal{L}(g)$ that we have just established we see that $\mathcal{L}$ maps $K$ into itself. In addition, we now show that
\begin{lemma} $\mathcal{L}[K]$ is pre-compact in $C([0,T];L^1(\Omega)).$
\end{lemma}
\begin{proof}
Indeed, since $f$ has bounded total variation in $p,$ we know that
\[
\sup_{t\in[0,T]}\int |f(p+h,t)-f(p,t)|\,dp{}\leq{}C(T)h.
\]
The support of functions $f(\cdot,t)$ for all different $t$'s and $g$'s is contained in some fixed ball $B_r$ because $X^t$ is an uniform translation with a continuous vector $\int_0^t\bb_g(\tau)\,d\tau.$ By Kolmogorov-Riesz-Frechet theorem, for all $t\in[0,T],$ set
\[
\left\{
\mathcal{L}(g)\right\}_{g\in K}
\]
is pre-compact in $L^1(\Omega).$ Using Lipschitz continuity in time, this also implies
that $\{\mathcal{L}\}_{g\in K}$ is pre-compact in $C([0,T];L^1(\Omega)).$
\end{proof}
Thus, $\mathcal{L}$ is a compact mapping from $K$ into itself.
By Schauder fixed point theorem, there is a fixed point $f=\mathcal{L}(f)$ in $K\subset C([0,T];L^1(\Omega)).$ Clearly, it verifies all estimates that we have derived. Moreover, it can be shown that $f$ is a weak solution of pde \eqref{FK:3}.

Uniqueness of solutions follows from a stronger property, stability estimate. Let $f_1,\,f_2$ be two solutions of \eqref{FK:3}--\eqref{BC} with initial conditions $f_{0,1}$ and $f_{0,2}.$ Such solutions verify the formula \eqref{f:formula}, from which we find that
\begin{multline*}
\int |f_1(p,t)-f_2(p,t)|\,dp{}\leq{}C(T)\int |f_{0,1}(p,t)-f_{0,2}(p,t)|\,dp\\
{}+{}C\int_0^t\int |f_1(p,\tau)-f_2(p,\tau)|\,dpd\tau.
\end{multline*}
Thus, according to Gronwall's inequality
\[
\int |f_1(p,t)-f_2(p,t)|\,dp{}\leq{}C(T)\int |f_{0,1}(p,t)-f_{0,2}(p,t)|\,dp.
\]
\end{proof}

Now we collect information on the support of solutions of \eqref{FK:3} that will  be used in the proof of theorem \ref{th:2}.

\begin{lemma}
\label{lemma:support}
Suppose that $\supp f_0\subset Interior(\Omega).$
Then, for any $t>0,$
\begin{enumerate}
\item[a.] $\supp f(\cdot,t)\subset Interior(\Omega);$
\item[b.] for any $p\in\Omega,$
\[
|p+\int_0^t\bb(\tau)\,d\tau|\geq t/d^2;
\]
\item[c.] if $\supp f_0\subset B_r(p_0),$ for some $r$ and $p_0\in\Omega,$ then
\[
\supp f(\cdot,t) \subset B_r\left(p_0+\int_0^t\bb(\tau)\,d\tau\right).
\]
\end{enumerate}
\end{lemma}
\begin{proof}
Since for any $t,$ $\bb(t)\in S_{d-1}\subset \Omega$ and $\Omega$ is a cone, 
it follows that $\int_0^t\bb(\tau)\,d\tau\in \Omega$ and for any $p\in Interior(\Omega),$
$p+\int_0^t\bb(\tau)\,d\tau\in Interior(\Omega).$ Moreover, the distance from $p+\int_0^t\bb(\tau),d\tau$ to $\partial\Omega$ is no less than the distance from $p$ to $\partial\Omega.$ This proves part. a.   Part b. follows from the fact that for any $t>0,$ $\sum_{i=1}^d\bb_i(t){}={}1,$ and so, there is $i_0,$ and there is $\delta_0\subset[0,t],$ such that $\bb_{i_0}(t)\geq 1/d,$ for all $t\in\Delta_0,$ and $|\Delta_0|\geq t/d.$ 
Part $c.$ follows immediately from \eqref{f:formula}. 
\end{proof}

\subsection{Asymptotic behavior in fictitious play}
Consider a model of statistical learning in a large population described by equation \eqref{SL:1}.
An initial boundary-value problem \eqref{IC}, \eqref{BC} with arbitrary $f_0\in C^1_0(\Omega),$ has a global unique solution, as was established in theorem \ref{th:1}. Denote population mean learning prior by
\[
P(t){}={}\int\frac{p}{\sum_i p_i}f(p,t)\,dp,
\]
and by $\hat{f}$ the projection of $f(p,t)$ onto the simplex $S_{d-1}.$ That is,
\[
\hat{f}(\hat{p},t){}={}f(p,t),\quad \hat{p}{}={}\frac{p}{\sum_i p_i}\in S_{d-1}.
\]
 The next theorem shows that if the population averages $P(t)$ and $\bb(t)$ converge to certain values, then these values must be the same and equal to a Nash equilibrium for the matrix game, and the learning priors of every agent in the population converge to that Nash equilibrium.
\begin{theorem}
\label{th:2}
Suppose that $\lim_{t\to\infty} P(t){}={}P_0$ and $\lim_{t\to\infty}\bb(t)=b_0.$ Then,   
\[
b_0=P_0\in \br(P_0),
\]
 and $\forall \epsilon>0,$ $\exists T(\epsilon)$ such that if $t>T(\epsilon),$ then
\begin{equation}
\label{converge}
\supp \hat{f}(\cdot,t)\subset B_\epsilon(P_0)\cap S_{d-1}.
\end{equation}
\end{theorem}
\begin{proof}
Consider function $\hat{f}(p,t)$ which is defined for $p\in S_{d-1}.$ From the definition 
of $P(t)$ it follows that $P(t)$ belongs to the closed convex hull spanned by $\supp \hat{f}(\cdot,t).$ At time $t=0,$ the support of $f_0$ is separated from the origin, and thus, by properties b. and c. of lemma \ref{lemma:support} (it applies to solutions of \eqref{SL:1} as well),
support of $f(\cdot,t)$ will be contained in a ball of fixed radius and the center diverging to infinity. This means that the diameter of the support of projection $\hat{f}$ decreases to zero. At the same time, since it contains  point $P(t)$ accumulating at $P_0,$ statement \eqref{converge} follows. 

To proof the first statement, notice that for sufficiently small $\epsilon$ and large $t,$
all of mass of $\hat{f}$ is near $P_0$ so that $\bb(t)$ is a convex combination of values of of $\br(p)$  in polytops adjacent to point $P_0,$ and so (see \eqref{br:cont} from Appendinx),  is an element of $\br(P_0).$ 
On the other hand $P_0$ must be equal to $b_0,$ because of the transport structure of the kinetic equation \eqref{SL:1}.
\end{proof}


\section{Examples}
\label{examples}
\subsection{Cultural evolution in Hawk-Dove-Retaliator  game}
\begin{table}
\centering
\begin{tabular}{@{}lccc@{}}
\toprule
    & Hawk           & Dove         & Retaliator \\
Hawk  & -1           &2           &   -1 \\
Dove  & 0            & 1           &  0.9\\
Retaliator  & -1          &  1.1         & 1\\
\bottomrule
\end{tabular}
\vspace{15pt}

\caption{Hawk-Dove-Retaliator game.  \label{HDR}}
\end{table}
Consider a classical Hawk-Dove-Retaliator game, table \ref{HDR}, from evolutionary game theory, see Maynard Smith \cite{Maynardbook} ans Zeeman \cite{Zeem}, table \ref{HDR}. The game  has two ESS: $\frac{1}{2}$Hawk+$\frac{1}{2}$Dove and Retaliator. Depending on the initial distribution of frequencies to  play hawk, dove, or retaliator, the replicator dynamics will proceed to one of ESS's as shown on figure \ref{fig:HDR1}. 

The same strategies are also asymptotically stable points for the best-response dynamics, which describes the statistical learning (fictitious play) in this game. Figure \ref{fig:HDR1}
shows the basins of attraction for each of the strategies.
Notice that basin of attraction for strategy R in replicator equation contains that region for the best-response dynamics.  The mismatch between two dynamics accounts for different scenarios of cultural learning for different pairs of the learning and reproduction rates $(\alpha_p,\alpha).$

For a population consisting of three groups, located in three best-respnse poligyons, figures \ref{fig:HDR2} and \ref{fig:HDR3} show two different scenarios for cultural evolution. The first is reproduction dominated and the other is learning dominated dynamics. Trajectories were obtained by solving \eqref{FK:3} numerically. 

 Notice also that the mean best-response (strategy profile) changes discontinuously when one the subpopulation crosses the boundaries of best response polygons.

With finite number of subpopulations the model reduces to a system of ODEs.
In this particular example the density function 
\[
f(p,t){}={}\sum_{i=1}^3w_i(t)\delta(p-p_i(t)),\quad \sum_{i=1}^3w_i(t)=1.
\]
where functions $p_i$ and $w_i$ are solutions of
\begin{eqnarray*}
\partial_t p_i{}&=&{}\bb(t),\,i=1..3,\\
\partial_t w_i{}&=&{}\alpha w_i\left(\sum_{kl}a_{kl}b_k(p_i(t))\bb_l(t) 
-a_{kl}\bb_k(t)\bb_l(t)
\right),\,i=1..3,
\end{eqnarray*}
and
\[
\bb(t){}={}w_1(t)b(p_1(t)){}+{}w_2(t)b(p_2(t)){}+{}w_3(t)b(p_3(t)).
\]
Notice that all priors $p_i(t)$ change in the direction of the mean best response $\bb(t)$ (when projected to $S_{d-1},$ this means that $p_i(t)$ moves toward $\bb(t)$), and the weights $w_i$ change according the performance of priors $p_i.$ 

Figures \ref{fig:HDR2} and \ref{fig:HDR3} were obtained using the following set of initial data: $w_1(0)=0.3,$ $w_2(0)=0.2,$ $w_3(0)=0.5,$ $p_1(0)=(0.1, 0.8, 0.1),$ $p_2(0)=(0.7, 0.2, 0.1),$ $p_3(0)=(0.05, 0.25, 0.7).$ With such initial data, the mean best response $\bb(0)=(0.3, 0.2, 0.5)$ is located in the basin of attraction of Retaliator according to the replicator equation and in the basin of attraction of $\frac{1}{2}$Hawk+$\frac{1}{2}$Dove for the best response dynamics. The values of $(\alpha,\alpha_p)$ are 
 $(5,1),$ for the example in figure \ref{fig:HDR2}, and $(3,1)$ in figure \ref{fig:HDR3}.

\begin{figure}[htbp]
\centering
\includegraphics[clip, trim=3.5cm 11cm 0.5cm 11.5cm,scale=0.9]{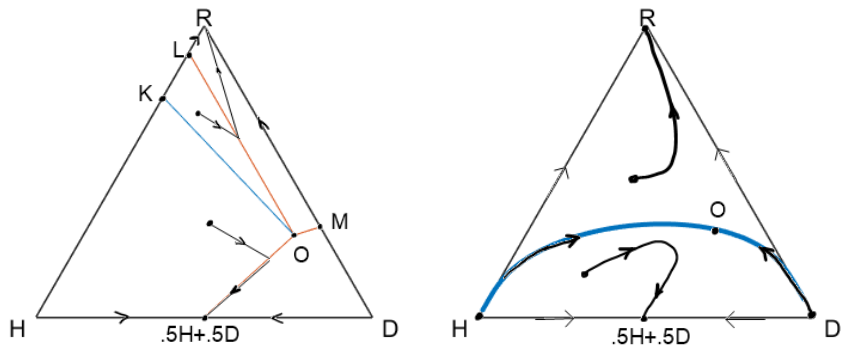}

\caption{Phase portraites for best-response (left) and replicator (right) equations for Hawk-Dove-Retaliator game in table \ref{HDR}.
On the left, three polygonal regions, formed by lines $OL,$ $OM$ and $O(.5H+.5D)$
are the regions where the best response function a single value: $H,$ $D,$ or $R.$  The basin of attraction for $R$ is polygon $KRMO$ (left) and
region above curve $HOD$ (right).  The plots show several trajectories for the best response and the replication equations. 
 \label{fig:HDR1}}
\end{figure}

\begin{figure}[htbp]
\centering
\includegraphics[clip, trim=0.5cm 9cm 0.5cm 9cm, width=1.00\textwidth]{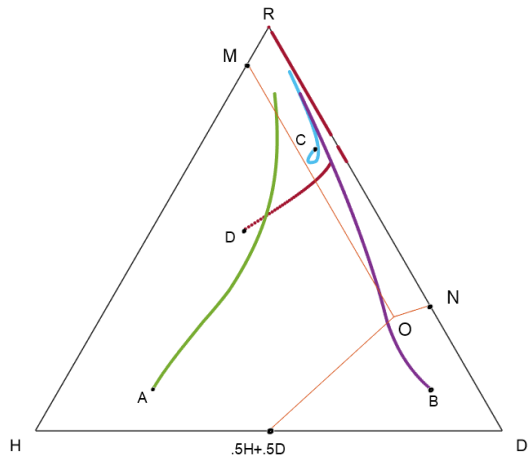}

\caption{Cultural evolution: reproduction dominating case. Three subpopulaitons starting at $A,$ $B,$ and $C,$ move toward Retaliator. A trajectory starting at $D$ is the mean best response (strategy profile). Notice that it changes discontinuously when one of the groups crosses  boundaries of best response polygons.
 \label{fig:HDR2} }
\end{figure}

\begin{figure}[htbp]
\centering
\includegraphics[clip, trim=0.5cm 9cm 0.5cm 9cm, width=1.00\textwidth]{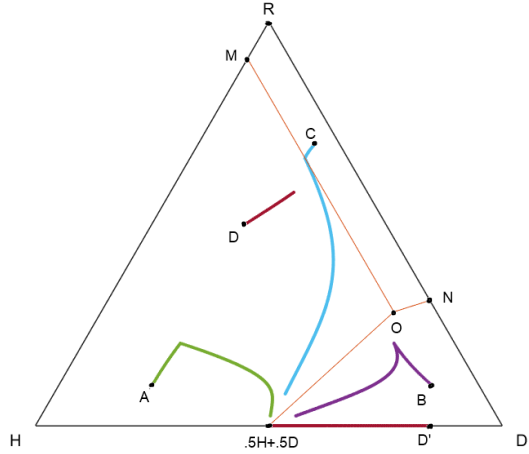}

\caption{Cultural evolution: learning dominating case. Same initial conditions as in Figure \ref{fig:HDR2}. All groups converge to .5H+.5D after the group that started at $C$ moves to the adjacent polygon. Trajectory starting at $D$ is the mean best response (strategy profile). Notice that it changes discontinuously (switches to $D'$) when the top group moves to the adjacent polygon.
 \label{fig:HDR3} }
\end{figure}

\subsection{Effect of growing population}

We consider a situation when the number of interactions among agents is constant and does not change if the population size $N(t)$ increases. That is, the effective learning rate
\[
\alpha_p{}={}\frac{2h}{N(t)\delta}{}={}\frac{N_0}{N(t)}\frac{2h}{N_0\delta}{}={}\frac{\alpha_1}{N(t)},
\]
where $N_0$ is the population size at time $t=0.$
The system of equations is 
\begin{equation}
\label{FK:RPS}
\partial_t f{}+{}\frac{\alpha_1}{N(t)}\div_p (\bb(t)f){}={}\alpha f\left(\sum_{ij}a_{ij}b_i(p)\bb_j(t){}-{}a_{ij}\bb_i(t)\bb_j(t)\right),
\end{equation}
\begin{equation}
\label{POP:RPS}
\partial_tN(t){}={}\alpha N\sum_{ij}\bb_i(t)\bb_j(t),
\end{equation}
with $\bb(t)$ given by \eqref{bb}.
Consider rock-paper-scissors game from table \ref{RPS}. We define the fitness levels (number of offspring) $a_{ij}$ as basis fitness 1 plus the numbers from the table.
The best response function $b(p)$ is sketched in figure \ref{fig:RPS}, to which we refer below.
Initially, the population is split into three groups. The first is 23/32 of all population and every agent in this groups has initially learning prior $p_1(0)=(0.5, 0.4, 0.1).$ It is located in the polygon RNOL for which the best response is to play ``paper''. The second group of proportion 1/4 has prior $p_2(0)=(0.2,0.7,0.1)$ with the best response ``scissors'',
and the third of proportion 1/32 has prior $p_3(0)=(0.32, 0.32, 0.36)$. The mean best response $\bb(0)$ is located on insider region $LONP.$  Suppose that initially there are 10 agents and the values of the parameters $\alpha_1 =1$ and $\alpha = 0.5.$
The distribution function $f$ has the form
\[
f(p,t){}={}\sum_{i=1}^3w_i(t)\delta(p-p_i(t)),\quad \sum_{i=1}^3w_i(t)=1,
\]
and the system \eqref{FK:RPS}, \eqref{POP:RPS} reduces to a system of 5 ODEs for $w_i(t),\,p_i(t),\,N(t),$ $i=1..3$

The dynamics of priors $p_1(t), $ $p_2(t),$ $p_3(t),$ and the mean best response $\bb(t)$ is shown on figure \ref{fig:RPS}, obtained from solving the system of ODEs numerically. In this dynamics,  statistical learning pushes $p_1,$ $p_2,$ $p_3$ toward $\bb(t),$
however, the rate of learning decreases exponentially (we show this below), and as the result $p_1(t),$ and $p_2(t)$ will asymptotically approach some locations in the same polygons where they have started, where as $p_3(t)$ moves to the decision polygon of $p_2(t)$ and also becomes locked there. Then,  the population frequency vector $\bb$ converges to ``scissors'' along line $PS,$ meaning that subpopulations that started in LOMP and MSNO out-evolves the first group.  

The dynamics here is different from that of the replicator equation for which $\bb(t)$ oscillates on a closed trajectory passing through the initial point $\bb(0).$ It differs also, from the dynamics of the best-response equation, that converges to the equilibrium $(1/3,1/3,1/3),$ see Gaunersdorfer and Hofbauer \cite{GH}.

To see that this scenario takes place, notice that   as long as $\bb(t)$ is located on line PS, $\bb_1(t)=0$ and we compute
\[
\partial_t N{}={}\alpha(\bb_2(t)+\bb_2(t))^2N{}={}\alpha N.
\]
Thus, the population grows exponentially, $N(t) = N_0e^{\alpha t}.$ In the state of priors $\mathbb{R}^3_+,$  each group moves to new positions given by formulas 
\[
p_i(t) = p_i(0){}+{}\frac{\alpha_1}{N_0}\int_0^t e^{-\alpha t}\bb(\tau)\,d\tau.
\]
Note, that figure \ref{fig:RPS} shows projections of this $p_i(t)$ onto $S_2.$
Clearly, $p_1(t)$ and $p_2(t)$ move a finite distance away from their initial position, and the parameters $\alpha, \alpha_1, N_0$ can be selected (as in this example) in such a way that $p_1(t)$ $p_2(t)$  remain in the polygon where it has started.  Moreover, a small fraction of population $w_3$ can be placed initially into region $OMSN,$ close to line $ON,$
so that it crosses that line, forcing $\bb(t)$ to move to line $PS.$ 
\begin{table}
\centering
\begin{tabular}{@{}lccc@{}}
\toprule
    &  Rock         & Paper           & Scissors  \\
Rock  & 0            & -1          &   1 \\
Paper  & 1            & 0           &  -1\\
Scissors  & -1          &  1           & 0\\
\bottomrule
\end{tabular}
\vspace{15pt}

\caption{Rock-Paper-Scissors game.  \label{RPS}}
\end{table}

\begin{figure}[htbp]
\centering
\includegraphics[clip, trim=0.5cm 9cm 0.5cm 9cm, width=1.00\textwidth]{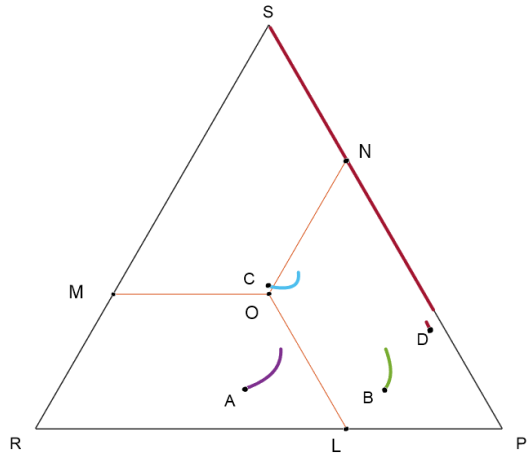}

\caption{Cultural evolution with constant interaction frequency. The plot shows three polygonal regions where the best response function takes a single value: $R,$ $P,$ or $S.$ There are three subpopulations located at points $A,B,$ and $C,$ respectively. The population starting at $C$ moves to the adjacent polygon and stays there for all subsequent times. Subpopulations started at $A$ and $B$ do not leave their polygons. A discontinuous trajectory starting in at $D$ is the mean best-response (strategy profile). Asymptotically it moves to $S,$ meaning that the subpopulations contained in the polygon $LONP,$  out-evolves the population from the adjacent polygon $LRMO.$
 \label{fig:RPS}}
\end{figure}

\end{section}

\begin{section}{Appendix}

\subsection{Best response function $\br(p).$} 
\label{best:response}
Let $S_{d-1}$ be the $d-1$ dimensional simplex $\left\{ p\in\mathbb{R}^d_+\,:\, \sum_i p_i{}={}1\right\}.$ Let $A{}={}\{a_{ij}\}$ represents payoff matrix in a symmetric game. We will assume that for no two indexes $i\not=j,$
\begin{equation}
\label{Game:hyp}
\sum_k a_{ik}p_k = \sum_k a_{jk}p_k,\quad \forall\, p\in S_{d-1}.
\end{equation} 
Denote by $r_i(p) = \sum_k a_{ik}p_k,$ the payoff to strategy $i$ played against mixed strategy $p,$  and a set
\[
\mathcal{I}(p) = \left\{ i_0(p)\in1..d\,:\,r_{i_0}(p) = \max_i r_i(p)\right\}.
\]
Denote the coordinate vectors $e_i{}={}(0,..0,1,0..0),$ with $1$ in $i^{th}$ position, and a multi-valued function
\begin{equation}
\label{BRF}
\br(p) {}={}\left\{\mbox{\rm convex hull of all $e_{i_0(p)},$ such that $i_0(p)\in \mathcal{I}(p)$}\right\}.
\end{equation}
Under hypothesis \eqref{Game:hyp}, $S_{d-1}$ is a union of finite number of polytops  such that $\br(p)$ is single-valued in the interior of each polytop $P_k,$ and at any point $p$ on the boundary of $P_k,$ the best response $\br(p)$ contains the value $\br(p_1)$ from the interior of $P_k:$
\[
\br(p_1)\in \br(p),\quad \forall p_1\in Interior(P_k),\, p\in \partial P_k.
\]
This condition can be re-phrased in an equivalent way, as a continuity condition:
for any $p\in S_{d-1},$ there is $\epsilon>0,$ such that 
for any $\epsilon_1<\epsilon,$ and any point $p_1\in B_{\epsilon_1}(p)\cap S_{d-1},$
\begin{equation}
\label{br:cont}
\br(p_1)\subseteq \br(p).
\end{equation}
Finally, we select a single-valued representative 
of  $b(p)$ from the values of $\br(p).$ If $p\in \mathbb{R}^d_+,$ then $b(p)$ is one of the values of $\br(p/\sum_i p_i).$ The selection can be, for example, the barycenter of the set of values of $\br(p),$ which corresponds to the situation when agents are choosing one strategy at random (from an uniform distribution). 
\end{section}

\bibliography{references_games}{}
\bibliographystyle{plain}

\end{document}